\documentclass[11pt]{article}
\usepackage[left=1in,top=1in,right=1in,bottom=1in,head=.1in,nofoot]{geometry}

\usepackage{setspace,url,bm,amsmath} % For double-spacing, URL font, math symbols

\usepackage{titlesec} % Section header formatting
\titlelabel{\thetitle.\quad} % Section header formatting
\titleformat*{\section}{\bf\large\center\uppercase} % Section header formatting

\usepackage{graphicx} % Graphics scaling
\usepackage{latexsym}
\usepackage{authblk}
\usepackage{hyperref}

\usepackage{amsthm}
\theoremstyle{definition}

\newtheorem{theorem}{Theorem}
\newtheorem{lemma}{Lemma}
\newtheorem{example}{Example}

\usepackage{natbib} % ASA citation style
\usepackage{etoolbox} % Bibliography underfull/overfull box fix
\apptocmd{\sloppy}{\hbadness 10000\relax}{}{} % Bibliography underfull/overfull box fix

\begin{document}
\doublespacing
\title{\bf Improved Neymanian analysis for $2^K$ factorial designs with binary outcomes}
\date{\today}
\author[1]{Jiannan Lu\thanks{Address for correspondence: Jiannan Lu, One Microsoft Way, Redmond, Washington 98052-6399, U.S.A. Email: \texttt{jiannl@microsoft.com}}}
\affil[1]{Analysis and Experimentation, Microsoft Corporation}
\maketitle
\begin{abstract}
$2^K$ factorial designs are widely adopted by statisticians and the broader scientific community. In this short note, under the potential outcomes framework \citep{Neyman:1923, Rubin:1974}, we adopt the partial identification approach and derive the sharp lower bound of the sampling variance of the estimated factorial effects, which leads to an ``improved'' Neymanian variance estimator that mitigates the over-estimation issue suffered by the classic Neymanian variance estimator by \cite{Dasgupta:2015}. 
\end{abstract}

\section{Introduction}
\label{sec:intro}

Originally introduced for agricultural research at the famous Rothamsted Experimental Station more than a century ago \citep{Fisher:1926, Yates:1963}, randomized controlled factorial designs \citep{Fisher:1935} have been widely adopted by researchers in social, behavior and biomedical sciences to simultaneously assess the main and interactive effects of multiple treatment factors. In applied research, a frequently encountered scenario is where not only the treatments but also the outcomes of interests are binary. For example, \cite{Nair:2008} explored how differentiating the format of a computer-aided encouragement program (e.g., single session vs. multiple occurrences, personalized vs. more general feedback and advice) affected the abstinence from smoking. \cite{Stampfer:1985} investigated whether aspirin and $\beta-$carotene could help prevent cardiovascular mortality. For such studies, to guarantee trustworthy discovery and reporting of causal effects that are scientifically meaningful, it is imperative to adopt an interpretable and robust methodology for estimation and inference. %Indeed, as pointed out by \cite{Brunner:2001}, ``the analysis of factorial designs is one of the most important and frequently encountered problems in statistics.'' 

During recent years, the potential outcomes framework \citep{Neyman:1923, Rubin:1974, Rubin:1990} has become increasingly popular, because it enjoys clear interpretability (causal effects are defined as comparisons between potential outcomes under different treatments), and can be flexibly combined with various inferential procedures (e.g., Fisherian, Neymanian and Bayesian; see \cite{Ding:2018} for a comprehensive review). Realizing the salient features of the potential outcomes framework, \cite{Dasgupta:2015} extended it to $2^K$ factorial designs, and claimed that the proposed Neymanian causal inference framework ``results in better understanding of the estimands and allows greater flexibility in statistical inference of factorial effects, compared to the commonly used linear model based approach.'' However, as acknowledged by the causal inference literature \citep[e.g.,][Section~6.5]{Aronow:2014, Ding:2017, Imbens:2015}, a long-standing and fundamental challenge faced by the Neymanian framework is the over-estimation of the sampling variances of the estimated factorial effects, because we cannot jointly observe the potential outcomes under different treatments, and therefore directly identify the strengths of association between them. This missing data problem is sometimes referred to as the ``fundamental problem of causal inference'' \citep[e.g.,][Section~1.3]{Holland:1986, Imbens:2015}.

Among the numerous proposals that mitigate the variance over-estimation of the Neymanian causal inference framework, one solution that completely preserves the randomization-based ``flavor'' is the partial identification approach \citep[c.f.][]{Richardson:2014}, which is widely employed by both statisticians \citep[e.g.,][]{Cheng:2006, Zhang:2003, Aronow:2014, Ding:2016, Lu:2018} and econometricians \cite[e.g.,][]{Fan:2010}. The key idea of partial identification, in the context of factorial designs, is that although we cannot directly identify the sampling variances of the estimated factorial effects, we can derive their sharp lower bounds which are identifiable from observed data, which leads to an ``improved'' Neymanian variance estimator that guarantees better performance, regardless of the underlying dependency structure of the potential outcomes (however, the extent of performance improvement depends on the dependency structure). Along this line of research, \cite{Ding:2016} solved the problem for treatment-control studies (i.e., $2^1$ factorial designs), and in a recent paper \cite{Lu:2017} proposed the said ``improved'' variance estimator for $2^2$ factorial designs. Nevertheless, we still need a unifying framework applicable to general $2^K$ factorial designs, which to the best of our knowledge is lacking from the existing literature. From a theoretical perspective, it seems non-trivial to generalize the main results in \cite{Lu:2017} to arbitrary $2^K$ factorial designs, because the complexity of the dependency structure of the potential outcomes grows exponentially as $K$ increases. From a practical perspective, although $2^2$ factorial designs seem common in applied research, high-order factorial designs were also frequently employed \citep{Berkowitz:1964, Kim:2008, Yuan:2008} (e.g., to screen a large number of candidate treatment factors). In this paper, we fill this theoretical gap by deriving the desired ``improved'' Neymanian variance estimator, for arbitrary $2^K$ factorial designs.

We organize the remainder of the paper as follows. Section \ref{sec:review} reviews \cite{Dasgupta:2015}'s Neymanian inference framework for $2^K$ factorial designs, with a primary focus on binary outcomes. Section \ref{sec:theory} first highlights the variance over-estimation issue that the Neymanian causal inference framework suffers, and presents the ``improved'' Neymanian variance estimator that is guaranteed to be less biased than the standard Neymanian estimator. Section \ref{sec:conclusion} concludes with a discussion.

\section{Neymanian inference for factorial designs}
\label{sec:review}

\subsection{Factorial designs}

We adapt some materials from \cite{Dasgupta:2015} and \cite{Lu:2016a, Lu:2016b} to review the Neymanian causal inference framework for $2^K$ factorial designs. To maintain consistency, we inherit the set of notations from \cite{Lu:2017}. 

Consider $K(\ge 1)$ distinct treatment factors with two-levels -1 (placebo) and 1 (active treatment), resulting a total number of $J=2^K$ treatment combinations, labelled as
$
\bm z_1, \ldots, \bm z_J.
$
Their definitions depend on the $J \times J$ model matrix $\bm H = (\bm h_0, \ldots, \bm h_{J-1})$ \citep[c.f.][]{Wu:2009}, constructed as follows: 
\begin{enumerate}

\item Let $\bm h_0 = \bm 1_J;$ 

\item For $k=1,\ldots,K$, construct $\bm h_k$ by letting its first $2^{K-k}$ entries be -1, the next $2^{K-k}$ entries be 1, and repeating $2^{k-1}$ times; 

\item If $K \ge 2,$ order all subsets of $\{1, \ldots, K\}$ with at least two elements, first by cardinality and then lexicography. For $k = 1, \ldots J-1-K,$ let $\sigma_k$ be the $k$th subset and $\bm h_{K+k} = \prod_{l \in \sigma_k} \bm h_l,$ where ``$\prod$'' stands for entry-wise product. 

\end{enumerate}
Given the resulting design matrix $\bm H,$ the $j$th row of the corresponding sub-matrix $(\bm h_1, \ldots, \bm h_K)$ is the $j$th treatment combination $\bm z_j,$ for $j = 1, \ldots, J.$ In the next section, we will use $(\bm h_1, \ldots, \bm h_J)$ to define the (main and interactive) factorial effects. 

\subsection{Potential outcomes and factorial effects}

Consider $2^K$ factorial designs with $N(\ge 2^{K+1})$ experimental units. Under the Stable Unit Treatment Value Assumption \citep[SUTVA,][]{Rubin:1980}, for all $i = 1, \ldots, N,$ we let $Y_i(\bm z_j)  \in \{0, 1\}$ be the potential outcome of unit $i$ under treatment $\bm z_j,$ and
$
\bm Y_i = \{ Y_i(\bm z_1), \ldots, Y_i(\bm z_J) \}^\prime.
$
To simplify future notations, for
$
j = 1, \ldots, J,
$
let
$
N_j 
= \sum_{i=1}^N 1_{\left\{Y_i(\bm z_{j}) = 1\right\}}
$
denote the umber of experimentation units with potential outcomes equal to one under treatment $\bm z_j.$ Similarly, for all
$
j, j^\prime = 1, \ldots, J
$
and
$
j \ne j^\prime,
$
let
$
N_{jj^\prime} 
= \sum_{i=1}^N 1_{\left\{Y_i(\bm z_j) = 1, Y_i(\bm z_{j^\prime}) = 1\right\}}
$
denote the umber of experimentation units with potential outcomes equal to one under both $\bm z_j$ and $\bm z_{j^\prime}.$ In other words, $N_j$'s and $N_{j, j^\prime}$'s determine the marginal and pair-wisely joint distributions of the potential outcomes. We will use this set of notations frequently going forward.

For all $j=1, \ldots, J,$ the average potential outcome for $\bm z_j$ is
$
p_j = N_j / N,
$
and
$
\bm p = ( p_1, \ldots, p_J )^\prime.
$
For all $l=1, \ldots, J-1,$ \cite{Dasgupta:2015} defined the $l$th individual- and population-level factorial effects as 
\begin{equation}
\label{eq:factorial-effects}    
\tau_{il} = 2^{-(K-1)} \bm h_l^\prime \bm Y_i
\;\;
(i = 1, \ldots, N);
\quad
\bar \tau_l = 2^{-(K-1)} \bm h_l^\prime \bm p.
\end{equation}

We provide the following example to illustrate the concepts introduced above.

\begin{example}
\label{example:1}
For $K=3,$ by the construction procedure described in Section 2.1, we obtain
$
\bm h_0 = (1, 1, 1, 1, 1, 1, 1, 1)^\prime,
$
$
\bm h_1 = (-1, -1, -1, -1, 1, 1, 1, 1)^\prime,
$
$
\bm h_2 = (-1, -1, 1, 1, -1, -1, 1, 1)^\prime,
$
and
$
\bm h_3 = (-1, 1, -1, 1, -1, 1, -1, 1)^\prime.
$
Moreover, 
$
\bm h_4 = \bm h_1 \cdot \bm h_2 = (1, 1, -1, -1, -1, -1, 1, 1)^\prime,
$
$
\bm h_5 = \bm h_1 \cdot \bm h_3,
$
$
\bm h_6 = \bm h_2 \cdot \bm h_3,
$
and
$
\bm h_7 = \bm h_1 \cdot \bm h_2 \cdot \bm h_3.
$
Therefore, the design matrix
\begin{equation*}
\bm H =
\bordermatrix{& \bm h_0 & \bm h_1& \bm h_2 & \bm h_3 & \bm h_4 & \bm h_5& \bm h_6 & \bm h_7 \cr
              & 1 & -1 & -1 & -1 &  1 &  1 &  1 & -1 \cr
              & 1 & -1 & -1 &  1 &  1 & -1 & -1 &  1 \cr
              & 1 & -1 &  1 & -1 & -1 &  1 & -1 &  1 \cr
              & 1 & -1 &  1 &  1 & -1 & -1 &  1 & -1 \cr
              & 1 &  1 & -1 & -1 & -1 & -1 &  1 &  1 \cr
              & 1 &  1 & -1 &  1 & -1 &  1 & -1 & -1 \cr
              & 1 &  1 &  1 & -1 &  1 & -1 & -1 & -1 \cr
              & 1 &  1 &  1 &  1 &  1 &  1 &  1 &  1
              },
\end{equation*}
and the treatment combinations are 
$$
\bm z_1=(-1, -1, -1),
\quad
\bm z_2=(-1, -1, 1),
\quad
\bm z_3=(-1, 1, -1),
\quad
\bm z_4=(-1, 1, 1),
$$
and
$$
\bm z_5=(1, -1, -1),
\quad
\bm z_6=(1, -1, 1),
\quad
\bm z_7=(1, 1, -1),
\quad
\bm z_8=(1, 1, 1),
$$
respectively. For illustration we consider the main effect of the first treatment factor. First, on the individual level, by \eqref{eq:factorial-effects}    
$$
\tau_{i1} 
= 
\frac{1}{4}\sum_{j=1}^4 Y_i(\bm z_j) 
- 
\frac{1}{4}\sum_{j=5}^8 Y_i(\bm z_j),
$$
which is indeed difference between the average potential outcome of unit $i$ where the first treatment factor is $-1,$ and the one where the first treatment factor is $+1.$ Second, 
\begin{align*}
\bar \tau_{i1} 
&= \frac{1}{N} \sum_{i=1}^N \tau_{i1} \\
&= \frac{1}{4N} 
\left\{
\sum_{i=1}^N \sum_{j=1}^4 Y_i(\bm z_j) - \sum_{i=1}^N \sum_{j=5}^8 Y_i(\bm z_j)
\right\}
 \\
&= \frac{1}{4N} 
\left\{
\sum_{j=1}^4 N_j - \sum_{j^\prime = 5}^8 N_{j^\prime}
\right\}.
\end{align*}
\end{example}

\subsection{Neymanian inference}

We consider a completely randomized treatment assignment. Let $n_1, \ldots, n_J$ be positive constants such that
$
\sum n_j = N.
$
For all $j = 1, \ldots, J,$ randomly assign $n_j \ge 2$ units to $\bm z_j.$ For all $i = 1, \ldots, N,$ we let
$$
W_i(\bm z_j) = 
\begin{cases}
1, \quad \text{if unit $i$ is assigned to $\bm z_j$,}\\
0, \quad \text{otherwise}
\end{cases}
\quad
(j = 1, \ldots, J).
$$
The observed outcomes for unit $i$ is therefore
$
Y_i^\textrm{obs} = \sum_{j=1}^J W_i(\bm z_j) Y_i(\bm z_j).
$
Let the average observed potential outcome for $\bm z_j$ be
$
\hat p_j = n_j^\mathrm{obs} / n_j,
$
where
\begin{equation*}
n_j^\textrm{obs} 
= \sum_{i=1}^N W_i(\bm z_j)Y_i(\bm z_j)
= \sum_{i: W_i(\bm z_j) = 1} Y_i^\textrm{obs}
\quad
(j = 1, \ldots, J).
\end{equation*}
Denote
$
\hat{\bm p}
= 
(
\hat p_1, 
\ldots,
\hat p_J
)^\prime.
$
An unbiased estimator of the factorial effect $\bar \tau_l$ is
\begin{equation}
\label{eq:factorial-effects-estimator}
\hat {\bar \tau}_l 
=  
2^{-(K-1)} \bm h_l^\prime \hat{\bm p}
\quad
(l = 1, \ldots, J-1).
\end{equation}
\cite{Dasgupta:2015} derived the sampling variance of the estimator in \eqref{eq:factorial-effects-estimator} as
\begin{equation}
\label{eq:factorial-effects-variance}
\mathrm{Var}(\hat {\bar \tau}_l) = \frac{1}{2^{2(K-1)}} \sum_{j=1}^J S_j^2 / n_j - \frac{1}{N} S^2(\bar \tau_l),
\end{equation}
where 
\begin{equation*}
S_j^2 
= 
(N-1)^{-1}
\sum_{i=1}^N 
\left\{ 
Y_i(\bm z_j) - \bar Y(\bm z_j) 
\right\}^2 
=
\frac{N}{N-1} p_j (1 - p_j)
\end{equation*}
is the variance of potential outcomes for $\bm z_j$, and 
$$
S^2(\bar \tau_l) 
= 
(N-1)^{-1}
\sum_{i=1}^N 
(\tau_{il} - \bar \tau_l)^2
$$
is the variance of the $l$th individual-level factorial effects. To estimate the sampling variance \eqref{eq:factorial-effects-variance}, \cite{Dasgupta:2015} substituted $S_j^2$ with its unbiased estimate
\begin{equation*}
s_j^2 
= 
(n_j - 1)^{-1} 
\sum_{i=1}^N
W_i(\bm z_j)
\{
Y_i^{\textrm{obs}} - \bar Y^{\textrm{obs}}(\bm z_j) 
\}^2
=
\frac{n_j}{n_j - 1} \hat p_j (1 - \hat p_j),
\end{equation*}
and plugged in the lower bound of zero for 
$
S^2(\bar \tau_l). 
$
The resulted Neymanian estimator
\begin{equation}
\label{eq:factorial-effects-variance-estimator}
\widehat{\mathrm{Var}}_{\mathrm{Ney}}(\hat {\bar \tau}_l) 
= 2^{-2(K-1)} \sum_{j=1}^J s_j^2 / n_j
= 2^{-2(K-1)} \sum_{j=1}^J \frac{\hat p_j (1 - \hat p_j)}{n_j - 1}
\end{equation}
is ``conservative,'' on average over-estimating the true sampling variance by
$$
\mathrm{E} 
\left\{ 
\widehat{\mathrm{Var}}_{\mathrm{Ney}}(\hat {\bar \tau}_l) 
\right\} 
- \mathrm{Var}(\hat{\bar \tau}_l) 
=  S^2(\bar \tau_l) / N.
$$
The bias is generally positive, unless strict additivity \citep{Dasgupta:2015, Ding:2016, Ding:2017} holds, i.e., 
$
\tau_{il} = \tau_{i^\prime l}
$
for all
$
i \ne i^\prime.
$
In other words, all experimental units have identical treatment effects. Several researchers \citep[e.g.,][]{Lavange:2005, Rigdon:2015} pointed out that this condition is too strong in practice, especially for binary outcomes. In cases where strict additivity does not hold, the estimator in \eqref{eq:factorial-effects-variance-estimator} might be too conservative, as acknowledged by \cite{Aronow:2014}.

\section{The improved Neymanian variance estimator}
\label{sec:theory}

The key to the partial identification approach is to derive a non-zero lower bound of $S^2(\bar \tau_l).$ To achieve this goal, we rely on the following lemmas, which are ``$2^K$ versions'' of the corresponding ``$2^2$ versions'' in \cite{Lu:2017}. However, it is worth mentioning that, the original proofs in \cite{Lu:2017} relied on the inclusion-exclusion principle and Boole's inequality, and therefore are difficult to be generalized to arbitrary $2^K$ factorial designs. To partially circumvent this technical difficulty, we adopt a methodology that is simpler and more intuitive than the one used by \cite{Lu:2017}. 

\begin{lemma}
\label{lemma:variance-taul-formula}
For all $l=1, \ldots, J-1,$ let 
$
\bm h_l = (h_{1l}, \ldots, h_{Jl})^\prime,
$
and
\begin{equation*}
S^2(\bar \tau_l) 
= 
\frac{1}{2^{2(K-1)}(N-1)} 
\left(
\sum_{j=1}^J N_j
+
\sum_{j \ne j^\prime} h_{lj}h_{lj^\prime} N_{jj^\prime}
\right)
-
\frac{N}{N-1} \bar \tau_l^2.
\end{equation*}
\end{lemma}

\begin{proof}
%[Proof of Lemma \ref{lemma:variance-taul-formula}]
The proof largely follows \cite{Lu:2017}. First, by \eqref{eq:factorial-effects}
\begin{align*}
\sum_{i=1}^N \tau_{il}^2 
& = 2^{-2(K-1)} \sum_{i=1}^N (\bm h_l^\prime \bm Y_i )^2 \\
& = 2^{-2(K-1)} \sum_{i=1}^N 
\left\{
\sum_{j=1}^J h_{lj}Y_i(\bm z_j)
\right\}^2 \\
& = 2^{-2(K-1)} \sum_{i=1}^N 
\left\{
\sum_{j=1}^J h_{lj}^2 Y^2_i(\bm z_j)
+
\sum_{j \ne j^\prime} h_{lj}h_{lj^\prime} Y_i(\bm z_j)Y_i(\bm z_{j^\prime})
\right\} \\
& = 2^{-2(K-1)}
\left\{
\sum_{j=1}^J h_{lj}^2 \sum_{i=1}^N  Y^2_i(\bm z_j)
+
\sum_{j \ne j^\prime} h_{lj}h_{lj^\prime} \sum_{i=1}^N Y_i(\bm z_j)Y_i(\bm z_{j^\prime})
\right\} \\
& = 2^{-2(K-1)} 
\left(
\sum_{j=1}^J N_j
+
\sum_{j \ne j^\prime} h_{lj}h_{lj^\prime} N_{jj^\prime}
\right).
\end{align*}
Therefore
\begin{align*}
S^2(\bar \tau_l) 
& = 
(N-1)^{-1}
\left(
\sum_{i=1}^N \tau_{il}^2 - N \bar \tau_l^2 \right) \\
& = \frac{1}{2^{2(K-1)}(N-1)} 
\left(
\sum_{j=1}^J N_j
+
\sum_{j \ne j^\prime} h_{lj}h_{lj^\prime} N_{jj^\prime}
\right)
-
\frac{N}{N-1} \bar \tau_l^2,
\end{align*}
which completes the proof.
\end{proof}

\medskip

\begin{lemma}
\label{lemma:ie-inequality}
For all $l = 1, \ldots, J-1,$ \begin{equation}
\label{eq:ie-inequality}
\sum_{j=1}^J N_j
+
\sum_{j \ne j^\prime} h_{lj}h_{lj^\prime} N_{jj^\prime}
\ge
\left|
\sum_{j=1}^J h_{lj} N_l
\right|,
\end{equation}
and the equality in \eqref{eq:ie-inequality} holds if and only if
\begin{equation}
\label{eq:ie-inequality-hold-1}
\tau_{il} 
\left\{ 
\tau_{il} + 2^{-(K-1)} 
\right\} 
= 0
\quad
(\forall i = 1, \ldots, N)
\end{equation}
or
\begin{equation}
\label{eq:ie-inequality-hold-2}
\tau_{il} 
\left\{ 
\tau_{il} - 2^{-(K-1)} 
\right\} 
= 0
\quad
(\forall i = 1, \ldots, N).
\end{equation}
\end{lemma}

\begin{proof}
%[Proof of Lemma \ref{lemma:ie-inequality}]
To prove \eqref{eq:ie-inequality}, we break it down into two parts:
\begin{equation}
\label{eq:ie-inequality-1}
\sum_{j=1}^J N_j
+
\sum_{j \ne j^\prime} h_{lj}h_{lj^\prime} N_{jj^\prime}
\ge
\sum_{j=1}^J h_{lj} N_l
\end{equation}
and
\begin{equation}
\label{eq:ie-inequality-2}
\sum_{j=1}^J N_j
+
\sum_{j \ne j^\prime} h_{lj}h_{lj^\prime} N_{jj^\prime}
\ge
- \sum_{j=1}^J h_{lj} N_l.
\end{equation}
Note that
\begin{enumerate}
\item The inequality in \eqref{eq:ie-inequality} holds \emph{if and only if both} the inequalities in \eqref{eq:ie-inequality-1} and \eqref{eq:ie-inequality-2} hold; 

\item The equality in \eqref{eq:ie-inequality} holds \emph{if and only if either} the equalities in both \eqref{eq:ie-inequality-1} and \eqref{eq:ie-inequality-2} holds. 

\end{enumerate}
We first prove \eqref{eq:ie-inequality-1}, and derive the sufficient and necessary condition for the equality to hold. To simplify notations, denote
$$
\bm J_{l-} = \{j: h_{lj} = -1\},
\quad
\bm J_{l+} = \{j: h_{lj} = 1\}.
$$
Simple algebra suggests that \eqref{eq:ie-inequality-1} is equivalent to
\begin{equation}
\label{eq:ie-inequality-1-1}
2\sum_{j \in \bm J_{l-}} N_j 
+ 
\sum_{j, j^\prime \in \bm J_{l-}; j \ne j^\prime} N_{jj^\prime}    
+ 
\sum_{j, j^\prime \in \bm J_{l+}; j \ne j^\prime} N_{jj^\prime}    
\ge
2\sum_{j \in \bm J_{l-}, j^\prime \in \bm J_{l+}} N_{jj^\prime}.   
\end{equation}
To prove \eqref{eq:ie-inequality-1-1}, for all $i = 1, \ldots, N,$ we let
$
\lambda_{il-} = \sum_{j \in \bm J_{l-}} Y_i(\bm z_j)
$
and 
$
\lambda_{il+} = \sum_{j^\prime \in \bm J_{l+}} Y_i(\bm z_{j^\prime}),
$
which are two integer constants. Therefore, for
$
i=1, \ldots, N,
$
it is obvious that 
$
(\lambda_{il-} - \lambda_{il+}) +
(\lambda_{il-} - \lambda_{il+})^2
\ge 0,
$
or equivalently
\begin{equation}
\label{eq:ie-inequality-1-2}
2\lambda_{il-} 
+
\lambda_{il-}(\lambda_{il-} - 1)
+
\lambda_{il+}(\lambda_{il+} - 1)
\ge 
2\lambda_{il-}\lambda_{il+}.
\end{equation}
Note that \eqref{eq:ie-inequality-1-2} immediately implies \eqref{eq:ie-inequality-1}, because
\begin{equation*}
\sum_{j \in \bm J_{l-}} N_j = \sum_{i=1}^N \lambda_{il-},
\quad
\sum_{j \in \bm J_{l-}, j^\prime \in \bm J_{l+}} N_{jj^\prime} = \sum_{i=1}^N \lambda_{il-}\lambda_{il+},
\end{equation*}
and
\begin{equation*}
\sum_{j, j^\prime \in \bm J_{ls}; j \ne j^\prime} N_{jj^\prime} = \sum_{i=1}^N \lambda_{ils}(\lambda_{ils} - 1)  
\quad
(s = -, +).   
\end{equation*}
Moreover, the equality in \eqref{eq:ie-inequality-1} holds if and only if the equality in \eqref{eq:ie-inequality-1-2} holds for all $i=1, \ldots, N,$ which is equivalent to \eqref{eq:ie-inequality-hold-1}, because by definition
$
\lambda_{il-} - \lambda_{il+} = 2^{K-1} \tau_{il}.
$

Similarly, we can prove \eqref{eq:ie-inequality-2}, and its equality holds if and only if \eqref{eq:ie-inequality-hold-2} holds.
\end{proof}

With the help of Lemmas \ref{lemma:variance-taul-formula} and \ref{lemma:ie-inequality}, along with the definition of factorial effect in \eqref{eq:factorial-effects}, we can derive the main theoretical result of the paper.

\begin{theorem}
\label{thm:variance-taul-lower-bound}
The sharp lower bound for 
$
S^2(\bar \tau_l)
$
is 
\begin{equation}
\label{eq:variance-taul-lower-bound}
S^2(\bar \tau_l)
\ge
\frac{N}{N-1}
\max 
\{
2^{-(K-1)}|\bar \tau_l| - \bar \tau_l^2, 0
\}.
\end{equation}
The equality in \eqref{eq:variance-taul-lower-bound} holds if and only if \eqref{eq:ie-inequality-hold-1} or \eqref{eq:ie-inequality-hold-2} holds.
\end{theorem}

\begin{proof}
%[Proof of Theorem \ref{thm:variance-taul-lower-bound}]
By Lemmas \ref{lemma:variance-taul-formula} and \ref{lemma:ie-inequality},
$$
S^2(\bar \tau_l) 
\ge
\frac{1}{2^{2(K-1)}(N-1)} 
\left|
\sum_{j=1}^J h_{lj} N_l
\right|
-
\frac{N}{N-1} \bar \tau_l^2.
$$
Moreover, by \eqref{eq:factorial-effects},
$$
N 2^{K-1} |\bar \tau_l| 
= 
N 
\left|
\sum_{j=1}^J h_{lj} p_j 
\right|
= 
\left|
\sum_{j=1}^J h_{lj} N_j
\right|,
$$
which completes the proof.
\end{proof}

\medskip

The lower bound in \eqref{eq:variance-taul-lower-bound} is ``sharp,'' in the sense that it is the minimum of all possible values of
$
S^2(\bar \tau_l)
$
compatible with the marginal distributions of the potential outcomes (i.e., $N_j$ for $j = 1, \ldots, J$). To facilitate a better understanding of Theorem \ref{thm:variance-taul-lower-bound}, we consider two special cases. First, when $K=1,$ we have the classic treatment-control studies. In this case, Theorem \ref{thm:variance-taul-lower-bound} reduces to the main result of \cite{Ding:2016}. Moreover, the condition in \eqref{eq:ie-inequality-hold-1} reduces to
$$
Y_i(1) = Y_i(-1) \;\: \mathrm{or} \;\: Y_i(-1) - 1
\quad
(\forall i=1, \ldots, N),
$$
or equivalently
$$
Y_i(1) \le Y_i(-1)
\quad
(i=1, \ldots, N),
$$
because $Y_i(1)$ and $Y_i(0)$ are both binary. Similarly, \eqref{eq:ie-inequality-hold-1} is equivalent to
$$
Y_i(1) \ge Y_i(-1)
\quad
(i=1, \ldots, N).
$$
The above two conditions are termed \emph{monotonicity} by \cite{Ding:2016}. Second, when $K=2,$ Theorem \ref{thm:variance-taul-lower-bound} reduces to the main result of \cite{Lu:2017}. 

We illustrate the results in Theorem \ref{thm:variance-taul-lower-bound} by the following numerical example. 

\begin{example}
\label{example:2}
We let
$$
S^2_\mathrm{lb}(\bar \tau_l) 
=
\frac{N}{N-1}
\max 
\{
2^{-(K-1)}|\bar \tau_l| - \bar \tau_l^2, 0
\}
$$
denote the lower bound for 
$
S^2(\bar \tau_l)
$
derived in Theorem \ref{thm:variance-taul-lower-bound}. We consider a balanced $2^3$ factorial design with 40 experimental units, where
\begin{itemize}

\item Case 1: For unit $i=1, \ldots, 20,$ let $\bm Y_i = (1, 1, 0, 0, 0, 0, 0, 1)$ and therefore $\tau_{i1} =-0.25.$ For unit $i=21, \ldots, 40,$ let $\bm Y_i = (0, 0, 1, 1, 1, 1, 0, 0)$ and therefore $\tau_{i1} =0.$ Consequently,
$
S^2(\bar \tau_1) = S^2_\mathrm{lb}(\bar \tau_1) = 0.016;
$

\item Case 2: For unit $i=1, \ldots, 36,$ let $\bm Y_i = (1, 1, 0, 0, 0, 0, 0, 1)$ and therefore $\tau_{i1} =-0.25.$ For unit $i=37, \ldots, 40,$ let $\bm Y_i = (0, 0, 1, 1, 1, 1, 0, 0)$ and therefore $\tau_{i1} =0.$ Consequently,
$
S^2(\bar \tau_1) = S^2_\mathrm{lb}(\bar \tau_1) = 0.0058.
$
\end{itemize}
We make two observations from the above examples. First, in both cases the lower bound $S^2_\mathrm{lb}(\bar \tau_1)$ is sharp. In other words, we can perfectly identify $S^2(\bar \tau_1),$ the heterogeneity in the individual factorial effects $\tau_{i1}$'s. This is because condition \eqref{eq:ie-inequality-hold-1} holds. Second, the extent to which we can improve upon the Neymanian variance estimator depends on $S^2(\bar \tau_1).$ Indeed, the larger the heterogeneity is, the larger the improvement can be.

\end{example}

Theorem \ref{thm:variance-taul-lower-bound} leads to the ``improved'' Neymanian variance estimator
\begin{equation}
\label{eq:factorial-effects-variance-estimator-improved}
\widehat{\mathrm{Var}}_{\mathrm{IN}}(\hat {\bar \tau}_l) 
= 
\widehat{\mathrm{Var}}_{\mathrm{Ney}}(\hat {\bar \tau}_l) 
-
\frac{1}{N-1}
\max 
\{
2^{-(K-1)}|\hat{\bar \tau}_l| - \hat{\bar \tau}_l^2, 0
\}
.
\end{equation}
This bias-correction term on the right hand side of \eqref{eq:factorial-effects-variance-estimator-improved} is always non-negative, implying a guaranteed improvement of variance estimation, for any observed data-set. %Moreover, it only depends on the estimated average factorial effect, implying that for two data-sets with different individual factorial effect heterogeneities but the same average factorial effect, the bias-correction will be similar. 

\section{Concluding remarks}
\label{sec:conclusion}

Under the potential outcomes framework, we have proposed an ``improved'' Neymanian variance estimator for $2^K$ factorial designs with binary outcomes. Comparing to the classic variance estimator by \cite{Dasgupta:2015}, the newly proposed estimator guarantees bias-correction, regardless of the underlying dependency structure of the potential outcomes. The core idea behind the new estimator is the sharp lower bound of the sampling variance of the estimated factorial effects. 

We point out two directions of future research. First, although we focus on binary outcomes, it would be interesting to generalize the current work to general outcomes (e.g., continuous, time to event). The proof of Lemma \ref{lemma:variance-taul-formula} suggest that the key is to sharply bound
$
\sum_{j \ne j^\prime} h_{lj}h_{lj^\prime} Y_i(\bm z_j)Y_i(\bm z_{j^\prime})
$
For $K=1,$ \cite{Aronow:2014} solved this problem by using the arrangement inequality \citep{Hardy:1988}. However, generalizing their results to factorial designs seems non-trivial, because there is no ``multivariate'' rearrangement inequality readily available, to the best our of knowledge. Second, in a recent paper \cite{Mukerjee:2018} extended the potential outcomes framework to more complex experimental designs beyond $2^K$ factorial (e.g., Latin square and split-plot), and it is possible to study partial identification for those scenarios. %Third, we can generalize our results beyond the classic ``difference-in-means'' type factorial effects, to more nuance causal estimands. 

\section*{Acknowledgement}

The author is grateful to the Editor and three anonymous reviewers for their valuable comments, which helped improve the quality and presentation of this paper significantly. The author thanks Prof. Tirthankar Dasgupta at Rutgers, Prof. Peng Ding at Berkeley and Dr. Yixuan Qiu at Carnegie Mellon, for insightful discussions.

\bibliographystyle{apalike}
\bibliography{factorial_binary}

\end{document}